\documentclass[11pt]{article}

\usepackage[margin=1in]{geometry}
\usepackage{amsmath,amssymb,amsthm,mathtools}
\usepackage{aliascnt}
\usepackage{microtype}
\usepackage[table,dvipsnames]{xcolor}
\usepackage{hyperref}
\usepackage[nameinlink,capitalise,noabbrev]{cleveref}

\hypersetup{
    colorlinks=true,
    linkcolor=blue,
    citecolor=ForestGreen,
    urlcolor=blue!60!black
}

\numberwithin{equation}{section}

\newtheorem{theorem}{Theorem}[section]
\newaliascnt{lemma}{theorem}
\newtheorem{lemma}[lemma]{Lemma}
\aliascntresetthe{lemma}
\newaliascnt{proposition}{theorem}

\aliascntresetthe{proposition}
\newaliascnt{corollary}{theorem}
\newtheorem{corollary}[corollary]{Corollary}
\aliascntresetthe{corollary}
\theoremstyle{definition}
\newaliascnt{definition}{theorem}
\newtheorem{definition}[definition]{Definition}
\aliascntresetthe{definition}
\newaliascnt{assumption}{theorem}
\newtheorem{assumption}[assumption]{Hypothesis}
\aliascntresetthe{assumption}
\newaliascnt{remark}{theorem}

\aliascntresetthe{remark}

\crefname{theorem}{Theorem}{Theorems}
\Crefname{theorem}{Theorem}{Theorems}
\crefname{lemma}{Lemma}{Lemmas}
\Crefname{lemma}{Lemma}{Lemmas}
\crefname{proposition}{Proposition}{Propositions}
\Crefname{proposition}{Proposition}{Propositions}
\crefname{corollary}{Corollary}{Corollaries}
\Crefname{corollary}{Corollary}{Corollaries}
\crefname{definition}{Definition}{Definitions}
\Crefname{definition}{Definition}{Definitions}
\crefname{assumption}{Hypothesis}{Hypotheses}
\Crefname{assumption}{Hypothesis}{Hypotheses}
\crefname{remark}{Remark}{Remarks}
\Crefname{remark}{Remark}{Remarks}
\crefname{section}{Section}{Sections}
\Crefname{section}{Section}{Sections}
\crefname{equation}{Equation}{Equations}
\Crefname{equation}{Equation}{Equations}

\DeclarePairedDelimiter\norm{\lVert}{\rVert}

\newcommand{\R}{\mathbb{R}}
\newcommand{\eps}{\varepsilon}
\newcommand{\one}{\mathbf{1}}
\newcommand{\supp}{\operatorname{supp}}

\newcommand{\E}{\mathbb{E}}
\newcommand{\SSE}{\mathsf{SSE}}

\title{The Condition-Number Barrier in Sparse Least Squares}
\author{}
\date{}

\author{
Honghao Lin\footnote{Google Research, Carnegie Mellon University / Texas A\&M University. \texttt{honghaol3010@gmail.com}}
\and
Vahab Mirrokni\footnote{Google Research. \texttt{mirrokni@google.com}}
\and
David P. Woodruff\footnote{Google Research and Carnegie Mellon University. \texttt{dpwoodru@gmail.com }}
}

\begin{document}
\maketitle

\begin{abstract}
In~\cite{axiotis2021sparse}, Axiotis and Sviridenko conjectured that the linear
dependence on the restricted condition number in sparse convex optimization
cannot be improved by a polynomial-time algorithm.  We establish their
conjectured lower bound for least-squares objectives, conditional on the
randomized exact-volume Small-Set Expansion Hypothesis in the weighted
regular-graph formulation of Raghavendra, Steurer, and
Tulsiani~\cite{raghavendra2012reductions}.
Concretely, for every fixed
$\gamma\in(0,1]$, there is no randomized polynomial-time algorithm that, with
probability at least $2/3$, returns a vector $x$ such that, writing
$s=\norm{x}_0$,
\[
    \norm{Ax-b}_2^2
    \leq \min_{\norm{z}_0\leq k}\norm{Az-b}_2^2+\eps
    \quad\text{and}\quad
    s=O\!\left(k\,\kappa_{s+k}^{\,1-\gamma}\right),
\]
where $\kappa_r$ is the restricted condition number at sparsity level $r$.
The result holds even on rational instances with $A$ of full column rank.

The proof was first obtained using a fully automated Gemini-based agentic system developed internally at Google. The authors have verified the proof and edited it for clarity of presentation.
\end{abstract}

\section{Introduction}
\label{sec:introduction}

Sparse convex optimization seeks a vector with few nonzero coordinates and
small objective value.  Given a target sparsity $k$ and an accuracy $\eps>0$,
the bicriteria goal is to attain an objective value no greater than the best
$k$-sparse value plus $\eps$, while allowing the output to use more than $k$
coordinates.  We study this question for the least-squares objective
\begin{equation}
    Q(x)=\norm{Ax-b}_2^2,
    \label{eq:least-squares-objective}
\end{equation}
Even the closely related task of finding the sparsest vector subject to a
prescribed residual tolerance is NP-hard~\cite{natarajan1995sparse}.  The
central question here is how much additional support suffices for an efficient
additive approximation.

The relevant structural parameter is the restricted condition number.  For
least squares, let $\kappa_r$ denote the ratio between the maximum and minimum
values of $\norm{Ah}_2^2$ over $r$-sparse unit vectors $h$.  Suppressing the
sparsity level in the notation, earlier analyses gave either a support bound
linear in $\kappa$ with an additional logarithmic dependence on the desired
accuracy~\cite{natarajan1995sparse,shalev2010trading}, or an
accuracy-independent $O(k\kappa^2)$
bound~\cite{shalev2010trading,jain2014iterative}.  Axiotis and Sviridenko
introduced Adaptively Regularized Hard Thresholding and obtained an
accuracy-independent $O(k\kappa)$ bound.  Their algorithm is randomized, and
the guarantee holds with high probability~\cite{axiotis2021sparse}.  They later
gave a deterministic adaptive-regularization method with the same
$O(k\kappa)$ support dependence and an improved
runtime~\cite{axiotis2022iterative}.

Axiotis and Sviridenko discuss in Section~5.1 of their
paper~\cite{axiotis2021sparse} a construction that appears in Appendix~B of the
full version of Foster, Karloff, and Thaler~\cite{foster2015variable}.  They
interpret it as an obstruction matching the linear dependence on $\kappa$ for
several greedy and hard-thresholding methods.  The construction is
algorithm-specific and does not rule out a different polynomial-time
algorithm.  Axiotis and Sviridenko therefore formulated the remaining question
as Conjecture~27~\cite{axiotis2021sparse}.  Writing $\gamma$ for the proposed
improvement in the condition-number exponent, the conjecture rules out any
polynomial-time algorithm that retains the additive objective guarantee while
returning a vector whose sparsity $s$ satisfies
\begin{equation}
    s=O\!\left(k\kappa_{s+k}^{1-\gamma}\right).
    \label{eq:conjectured-sparsity}
\end{equation}
The occurrence of $s$ on both sides is intentional: the restricted condition
number is evaluated at the combined target and output support sizes.

\subsection{Our result}

We establish this general computational barrier conditional on the randomized
exact-volume Small-Set Expansion Hypothesis in the weighted regular-graph
formulation of Raghavendra, Steurer, and
Tulsiani~\cite{raghavendra2012reductions}, as formalized in
\cref{hyp:exact-sse}.  The conclusion covers arbitrary randomized
polynomial-time algorithms.  The hardness holds even when $A$ and $b$ have
rational entries and $A$ has full column rank.

The hypothesis asks one to distinguish a regular graph containing a
low-expansion set of one fixed volume from a graph in which every set of
exactly that volume has expansion close to one.  This exact-volume YES/NO
structure is the $\SSE(\eta,\delta)$ formulation of Raghavendra, Steurer, and
Tulsiani~\cite{raghavendra2012reductions}, where the prescribed cardinality is
$k=\delta n$.  Our hypothesis
considers randomized hardness on rational weighted regular graphs.  The
exact-size character is central to the proof: the core argument derives the
needed control over arbitrary returned supports from expansion at the single
target volume.

The formal statement below makes the hidden constant in
\cref{eq:conjectured-sparsity} explicit.  As usual, that constant may depend
on the algorithm and on the fixed exponent $\gamma$, but not on the input.

\begin{theorem}[Main theorem]
\label{thm:main-intro}
Assume the randomized weighted exact-volume Small-Set Expansion Hypothesis
stated in \cref{hyp:exact-sse}.  Fix constants $\gamma\in(0,1]$ and $C\geq1$.
There is no randomized algorithm, with running time polynomial in the binary
encoding length of $(A,b,k,\eps)$ and in $1/\eps$, that, given
$A\in\mathbb{Q}^{m\times n}$, $b\in\mathbb{Q}^m$,
$k\in\{1,\ldots,n\}$, and $\eps\in\mathbb{Q}_{>0}$, with $A$ of full column
rank, returns $x\in\mathbb{Q}^n$ such that, with probability at least $2/3$,
writing $s=\norm{x}_0$, both
\begin{align}
    \norm{Ax-b}_2^2
    &\leq \min_{\norm{z}_0\leq k}\norm{Az-b}_2^2+\eps,
    \label{eq:main-objective-guarantee}\\
    s&\leq Ck\,
        \kappa_{s+k}^{1-\gamma}.
    \label{eq:main-support-guarantee}
\end{align}
Consequently, conditional on this hypothesis, the algorithmic guarantee ruled
out by Conjecture~27 of Axiotis and Sviridenko is impossible for every fixed
$0<\gamma\leq1$.
Moreover, for each fixed $(C,\gamma)$, there is a constant
$K_{C,\gamma}<\infty$ such that the hard instances may be chosen to satisfy
\[
    \kappa_r(A)\leq K_{C,\gamma}
    \qquad\text{for every }r,
\]
and every row of $A$ has at most two nonzero entries.
\end{theorem}

The feasibility premise in Conjecture~27~\cite{axiotis2021sparse} is automatic
for the instances in our reduction.  Because $A$ has full column rank, $Q$ is
coercive and attains its
minimum over the finite union of coordinate subspaces
$\{z:\norm{z}_0\leq k\}$.  Let $z^\star$ be such a minimizer and write
$r^\star=\norm{z^\star}_0$.  Then
\[
    r^\star\leq k\leq
    Ck\kappa_{r^\star+k}^{1-\gamma}.
\]
The second inequality follows from $C\geq1$, $\kappa_{r^\star+k}\geq1$, and
$1-\gamma\geq0$.  Thus $z^\star$ itself satisfies both guarantees.


\subsection{Proof overview}
\label{sec:overview}

Let $G$ be a rational weighted regular graph with $n$ vertices.  Normalize its
weighted degree to one, and let $L_G$ be its Laplacian.  We choose a constant
$c>0$ later as a function only of the fixed pair $(C,\gamma)$ and set $\mu=c$.
Consider
\begin{equation}
    F_G(x)=x^\top L_Gx+\mu\norm{x}_2^2-2\one^\top x.
    \label{eq:overview-quadratic}
\end{equation}
Here $\one\in\R^n$ is the all-ones vector.  To expose the least-squares
structure, first work over $\R$.  For each non-loop edge $\{u,v\}$ of weight
$w_{uv}$, add the edge row $\sqrt{w_{uv}}(e_u-e_v)^\top$ with response zero.
For each vertex $v$, add the ridge row and response
\[
    \sqrt{\mu}\,e_v^\top,
    \qquad \frac1{\sqrt{\mu}},
\]
respectively.  Thus every edge row has two nonzero entries and every ridge row
has one.  Expanding the squared residuals gives
\begin{equation}
\begin{aligned}
    \norm{Ax-b}_2^2
    &=\sum_{u<v}w_{uv}(x_u-x_v)^2
      +\sum_{v\in V}\left(\sqrt{\mu}\,x_v-\frac1{\sqrt{\mu}}\right)^2 \\
    &=x^\top L_Gx+\mu\norm{x}_2^2-2\one^\top x+\frac n\mu
      =F_G(x)+\frac n\mu.
\end{aligned}
\label{eq:overview-least-squares}
\end{equation}
If $B_G$ denotes the block of edge rows, then
$B_G^\top B_G=L_G$, while the ridge rows contribute $\mu I$.  Therefore
\[
    A^\top A=L_G+\mu I.
\]
Since $L_G$ is positive semidefinite and $\mu>0$, this matrix is positive
definite, so $A$ has full column rank.  Moreover, the spectrum of the
Laplacian of a weighted graph of degree one lies in
$[0,2]$~\cite{chung1997spectral}.  Thus, for every $r$ and every $r$-sparse
unit vector $h$,
\[
    \mu
    \leq \norm{Ah}_2^2
    =h^\top(L_G+\mu I)h
    \leq 2+\mu.
\]
By the definition of $\kappa_r$, these bounds imply
\[
    \kappa_r\leq\frac{2+\mu}{\mu}=1+\frac2c
\]
for every sparsity level $r$.  The square roots above need not be rational.
In \cref{lem:rational-realization}, each edge or ridge coefficient is instead
realized as a sum of squares of rationals.  The resulting rows have the same
supports, so this replacement preserves $\norm{Ax-b}_2^2$, $A^\top A$, and the
property that every row has at most two nonzero entries.

We next establish a gap between the two cases of the SSE instance.  Write
$\Phi_G(S)$ for the normalized edge expansion of $S$.  The parameter $\eta$
specifies the SSE promise: in the YES case, $G$ contains a set $S^\star$ of
size $k$ with $\Phi_G(S^\star)\leq\eta$, whereas in the NO case every set of
size $k$ has expansion at least $1-\eta$.  In the YES direction
(completeness), consider the $k$-sparse vector $\alpha\one_{S^\star}$, where
$\one_{S^\star}$ is the indicator of $S^\star$.  Since
$\one_{S^\star}^\top L_G\one_{S^\star}
=k\Phi_G(S^\star)\leq k\eta$ and $\mu=c$,
\[
    F_G(\alpha\one_{S^\star})
    \leq k\bigl((\eta+c)\alpha^2-2\alpha\bigr).
\]
The right-hand side is minimized at $\alpha=1/(\eta+c)$, giving
\[
    \min_{\norm{x}_0\leq k}F_G(x)
    \leq -\frac{k}{\eta+c}.
\]

In the NO direction (soundness), assume that $G$ is a NO instance and let
$S=\supp(x)$ be the support returned on a run in which the hypothetical
algorithm satisfies its guarantees.  The support guarantee and the uniform
condition-number bound imply
\[
    |S|
    \leq Ck\kappa_{|S|+k}^{1-\gamma}
    \leq Ck(1+2/c)^{1-\gamma}
    \leq Lk
\]
for a suitable constant $L$.  We prove a lower bound on
$\min_{\supp(y)\subseteq S}F_G(y)$ for every such $S$.  The difficulty is that
$|S|$ need not equal $k$ and may be larger, whereas the exact-size SSE NO
promise applies only to sets of size exactly $k$.  It therefore cannot be
applied directly to $S$.

Fix a nonempty $S$ and take $\eta=\theta^2$, with $\theta>0$ chosen
sufficiently small later.  We decompose $S$ into a nonexpanding core $U$ and
an expanding residual $W$.
Initialize $U=\varnothing$ and $R=S$.  Whenever some nonempty $X\subseteq R$
satisfies
\[
    w_G(X,R\setminus X)+w_G(X,V\setminus S)
    <(1-\theta)|X|,
\]
move $X$ from $R$ to $U$.  At termination, set $W=R$.  Summing the removal
inequalities shows that $\Phi_G(U)<1-\theta$ when $U\neq\varnothing$, while
the absence of another removable set means that $W$ has modified Dirichlet
expansion at least $1-\theta$ when $W\neq\varnothing$.

The exact-size NO promise nevertheless forces $U$ to be small.  On the
sufficiently large instances used in the reduction, $k\geq2$.  Suppose first
that $|U|>k$.  For a uniformly random $k$-element subset $T$ of $U$, the low
expansion of $U$ and the bound $|U|\leq Lk$ give
\[
    \E\Phi_G(T)<1-\frac{\theta}{2L}
                <1-\theta^2=1-\eta,
\]
where the second inequality follows by choosing $\theta<1/(2L)$.  This
contradicts the NO promise, so $|U|\leq k$.  When $U\neq\varnothing$, pad $U$
to a set $T$ of size exactly $k$.  The added vertices contribute at most one
unit of boundary weight each, and hence
\[
    1-\eta
    \leq\Phi_G(T)
    \leq\frac{|U|\Phi_G(U)+k-|U|}{k}.
\]
Rearranging gives
\[
    |U|\bigl(1-\Phi_G(U)\bigr)\leq \eta k.
\]
Since $\eta=\theta^2$ and $1-\Phi_G(U)>\theta$, the preceding inequality gives
$|U|<\theta k$.  When $W\neq\varnothing$, a Dirichlet Cheeger
inequality~\cite{chung2007random} turns its modified expansion into the
spectral bound
\[
    \lambda_{\min}(L'_W)\geq \frac12(1-\theta)^2,
\]
where $L'_W$ retains the edges inside $W$ and the boundary contributions from
$W$ to $V\setminus S$, but not the edges between $U$ and $W$.

Deleting these edges can only decrease the Laplacian quadratic form and
decouples the coordinates in $U$ from those in $W$.  For a positive-definite
matrix $M$,
\[
    \min_y\bigl(y^\top My-2\one^\top y\bigr)
    =-\one^\top M^{-1}\one.
\]
Applying this identity to the two blocks, using the fact that the core block is
bounded below by $\mu I$ and that the residual block is bounded below by
$(\lambda_{\min}(L'_W)+\mu)I$, gives the following bound, with the $W$-term
omitted when $W=\varnothing$:
\[
\begin{aligned}
    \min_{\supp(y)\subseteq S}F_G(y)
    &\geq
    -\frac{|U|}{\mu}
    -\frac{|W|}{\lambda_{\min}(L'_W)+\mu}\\
    &\geq
    -\frac{\theta k}{c}
    -\frac{Lk}{\tfrac12(1-\theta)^2+c}.
\end{aligned}
\]
The second line uses $|U|<\theta k$, $|W|\leq Lk$, and $\mu=c$.

Finally, \cref{lem:parameter-selection} chooses $c,L,\theta$ so that
\[
    C(1+2/c)^{1-\gamma}\leq L,
    \qquad \theta<\frac1{2L},
    \qquad \Gamma(c,L,\theta)>0,
\]
where $\Gamma$ is the difference between the completeness and soundness
coefficients above.  The resulting least-squares gap is
$\Delta_G=k\Gamma(c,L,\theta)$.  Calling the hypothetical algorithm with
$\eps_G=\Delta_G/3$ would distinguish the YES and NO cases with probability at
least $2/3$, contradicting \cref{hyp:exact-sse}.

\subsection{Related work}

For sparse convex optimization, early algorithms either retained a dependence
on the requested accuracy or gave accuracy-independent support bounds of order
$k\kappa^2$~\cite{natarajan1995sparse,shalev2010trading,jain2014iterative}.
Axiotis and Sviridenko introduced Adaptively Regularized Hard Thresholding and
obtained an $O(k\kappa)$ support bound with high probability, removing both the
accuracy dependence and the extra factor of
$\kappa$~\cite{axiotis2021sparse}.  Their later deterministic variant retains
the $O(k\kappa)$ support guarantee, improves the runtime, and does not require
prior knowledge of the optimal value or sparsity~\cite{axiotis2022iterative}.

Classical worst-case hardness results for sparse regression are not
parameterized by the restricted condition number.  Natarajan proved that it is
NP-hard to find the sparsest vector attaining a prescribed residual
tolerance~\cite{natarajan1995sparse}.  Foster, Karloff, and Thaler introduced a
bicriteria $(g,h)$-sparse-regression formulation that simultaneously relaxes
support and prediction error~\cite{foster2015variable}.  For an $m\times p$
design, under the assumption
$\mathrm{NP}\nsubseteq
\mathrm{BPTIME}(N^{\operatorname{polylog} N})$, where $N$ denotes the input
length, their main result shows that for any fixed positive constants
$\delta,C_1,C_2$ there are functions
$g(p)\in 2^{\Omega(\log^{1-\delta}p)}$ and
$h(m,p)\in\Omega(p^{C_1}m^{1-C_2})$ for which no quasipolynomial-time
randomized algorithm exists.  Under the Biregular Projection Games Conjecture
together with $\mathrm{NP}\nsubseteq\mathrm{BPP}$, they obtain analogous
hardness against randomized polynomial-time algorithms with
$g(p)\in p^{\Omega(1)}$, and they establish corresponding results for a noisy
statistical formulation.  These theorems give strong general bicriteria
inapproximability, but do not calibrate the hardness to $\kappa$.

Separate from those complexity-theoretic results, Appendix~B of the full
version of Foster, Karloff, and Thaler gives a concrete bad instance for
forward stepwise selection~\cite{foster2015variable}.  In Section~5.1,
Axiotis and Sviridenko interpret this construction as a linear-in-$\kappa$
barrier for OMP, OMPR, IHT, and Partial Hard Thresholding, and note that the
same example applies to ARHT and Exhaustive Local
Search~\cite{axiotis2021sparse}.  Their discussion also lists LASSO; the Foster,
Karloff, and Thaler full version itself states that the example thwarts LASSO
but does not include that proof.  For the OMP-like methods, the mechanism is
explicit: the gradient criterion, when initialized outside the optimal
support, favors distractor coordinates and never enters that support.  This is
an algorithm-specific obstruction, rather than a complexity lower bound for
arbitrary polynomial-time algorithms.

Conjecture~27 of Axiotis and Sviridenko asks whether the exponent of $\kappa$
can be improved by any polynomial-time algorithm.  Our result addresses this
remaining question for least squares: conditional on \cref{hyp:exact-sse}, it
rules out the support bound $O(k\kappa_{s+k}^{1-\gamma})$ for arbitrary
randomized polynomial-time algorithms.
Thus the lower bound is directly parameterized by the same restricted
condition number that appears in the upper bounds.

The Small-Set Expansion Hypothesis arose from work relating graph expansion to
the Unique Games Conjecture~\cite{raghavendra2010graph}.  Reductions among
exact-volume and robust expansion problems were developed by Raghavendra,
Steurer, and Tulsiani~\cite{raghavendra2012reductions}.  The hypothesis used
here concerns randomized hardness of this exact-volume problem on rational
weighted regular inputs.  Our main soundness step converts expansion at one
prescribed cardinality into a bound on the nonexpanding core of an arbitrary
relaxed support.

\section{Preliminaries}
\label{sec:preliminaries}

Throughout, numerical inputs and outputs are rational and encoded in binary.

\subsection{Expansion and randomized exact-volume SSEH}

Let $G$ be an undirected weighted graph on $V=[n]$, represented by symmetric
nonnegative rational weights $w_{uv}=w_{vu}$, with its nonzero weights listed
and encoded in binary.  We allow self-loops.  The weighted degree of every
vertex is normalized to one:
\begin{equation}
    \sum_{v\in V}w_{uv}=1
    \qquad\text{for every }u\in V.
    \label{eq:weighted-regularity}
\end{equation}
This normalization loses no generality: if the common weighted degree is $d$,
dividing every weight by $d$ converts the usual conductance
$w_G(S,V\setminus S)/(d|S|)$ into the degree-one expression below and preserves
polynomial bit complexity.  For sets $S,T\subseteq V$, write
\[
    w_G(S,T)=\sum_{u\in S}\sum_{v\in T}w_{uv},
    \qquad
    I_G(S)=w_G(S,S).
\]
For a nonempty set $S\subseteq V$, its normalized edge expansion is
\begin{equation}
    \Phi_G(S)=\frac{w_G(S,V\setminus S)}{|S|}.
    \label{eq:edge-expansion}
\end{equation}
Regularity gives
\begin{equation}
    \Phi_G(S)=1-\frac{I_G(S)}{|S|}.
    \label{eq:expansion-internal-edges}
\end{equation}
Let $W_G=(w_{uv})_{u,v\in V}$ and $L_G=I-W_G$.  Then
\begin{equation}
    x^\top L_Gx
    =\frac12\sum_{u,v\in V}w_{uv}(x_u-x_v)^2
    =\sum_{u<v}w_{uv}(x_u-x_v)^2.
    \label{eq:laplacian-form}
\end{equation}
In particular, for every nonempty $S$,
\begin{equation}
    \one_S^\top L_G\one_S
    =w_G(S,V\setminus S)
    =|S|\Phi_G(S).
    \label{eq:indicator-laplacian}
\end{equation}

The exact-volume YES/NO structure below is the
$\SSE(\eta,\delta)$ promise of Raghavendra, Steurer, and
Tulsiani~\cite{raghavendra2012reductions}; equivalently, it may be stated using
the uniform vertex measure $\nu(S)=|S|/n$.  We state its randomized
finite-precision version, with rational edge weights encoded in binary.  The
usual NP-hardness formulation together with $\mathrm{NP}\nsubseteq\mathrm{BPP}$
implies this randomized formulation.

\begin{assumption}[Randomized weighted exact-volume SSE hypothesis]
\label{hyp:exact-sse}
For every constant $\eta\in(0,1/2)$, there exists a constant
$\delta_{\mathrm{SSE}}\in\mathbb{Q}\cap(0,\eta]$ such that the following
promise problem cannot be solved by a randomized polynomial-time algorithm with success
probability at least $2/3$ on every promised input.  Given an undirected,
weighted graph $G$ on $n$ vertices satisfying \cref{eq:weighted-regularity},
with rational weights encoded in binary, and restricted to admissible values
of $n$ for which
$k=\delta_{\mathrm{SSE}}n$ is an integer, distinguish between
\begin{align*}
    \textnormal{YES:}\quad
    &\textnormal{there exists }S^\star\subseteq V,
      \ |S^\star|=k,
      \ \Phi_G(S^\star)\leq\eta;\\
    \textnormal{NO:}\quad
    &\textnormal{every }S\subseteq V,
      \ |S|=k,
      \ \Phi_G(S)\geq1-\eta.
\end{align*}
The hard family contains arbitrarily large admissible values of $n$.
\end{assumption}

The restriction $\delta_{\mathrm{SSE}}\leq\eta$ entails no loss: for a
uniformly random $k$-element set $T$, one has
$\E\Phi_G(T)\leq(n-k)/(n-1)$, so arbitrarily large NO instances require
$\delta_{\mathrm{SSE}}\leq\eta$.

The equality $|S|=k$ is part of both sides of the promise: in particular, the
NO case makes no assertion about sets of smaller or larger cardinality.  This
exact-size restriction is used in \cref{lem:core-size}.  The probability in
\cref{hyp:exact-sse} is over the internal randomness of the distinguishing
algorithm; the constant $2/3$ may be replaced by any fixed constant greater
than $1/2$ by standard amplification.

\subsection{Restricted condition numbers}

Let $A\in\R^{m\times n}$ have full column rank.  For an integer $r\geq1$,
define
\begin{equation}
    \kappa_r(A)
    =
    \frac{\displaystyle
        \max_{\substack{h\in\R^n:\ \norm{h}_2=1,\ \norm{h}_0\leq r}}
        \norm{Ah}_2^2}
         {\displaystyle
        \min_{\substack{h\in\R^n:\ \norm{h}_2=1,\ \norm{h}_0\leq r}}
        \norm{Ah}_2^2}.
    \label{eq:restricted-condition-number}
\end{equation}
For $r\geq n$, we use the convention $\kappa_r(A)=\kappa_n(A)$.  The
denominator is positive because $A$ has full column rank.  If $A_T$ denotes
the submatrix formed by the columns indexed by $T$, then the maximum
and minimum of $\norm{Ah}_2^2$ over unit vectors supported on $T$ are
$\lambda_{\max}(A_T^\top A_T)$ and
$\lambda_{\min}(A_T^\top A_T)$, respectively.  Thus
\cref{eq:restricted-condition-number} is equivalent to the usual
principal-submatrix definition.  In the normalization of Axiotis and
Sviridenko, the restricted smoothness and strong-convexity constants are twice
the numerator and denominator of \cref{eq:restricted-condition-number} for the
least-squares objective $\norm{Ax-b}_2^2$, respectively.  The factor of two
cancels, so this is exactly the condition number used in
Conjecture~27~\cite{axiotis2021sparse}.  When the matrix $A$ is understood, we
write $\kappa_r$ for $\kappa_r(A)$.

For later reference we isolate the algorithmic guarantee that the reduction
rules out.

\begin{definition}[Condition-number-improving sparse least squares]
\label{def:algorithm}
Fix constants $\gamma\in(0,1]$ and $C\geq1$.  A $(C,\gamma)$ sparse
least-squares algorithm is a randomized algorithm that, on every input
$A\in\mathbb{Q}^{m\times n}$, $b\in\mathbb{Q}^m$,
$k\in\{1,\ldots,n\}$, and $\eps\in\mathbb{Q}_{>0}$, with $A$ of full column
rank, returns $x\in\mathbb{Q}^n$ for which the two guarantees
\begin{equation}
\begin{aligned}
    \norm{Ax-b}_2^2
    &\leq\min_{\norm{z}_0\leq k}\norm{Az-b}_2^2+\eps,\\
    \norm{x}_0
    &\leq Ck\kappa_{\norm{x}_0+k}(A)^{1-\gamma}.
\end{aligned}
    \label{eq:algorithm-definition}
\end{equation}
hold simultaneously with probability at least $2/3$ over the internal
randomness of the algorithm.
\end{definition}

The index $\norm{x}_0+k$ in \cref{eq:algorithm-definition} depends on the
output itself.  The reduction handles this self-reference by bounding
$\kappa_r(A)$ uniformly over all $r$.

\section{Graph-to-regression construction}
\label{sec:reduction}

Fix a positive rational constant $c$; its final value will be chosen in
\cref{sec:parameters}.

\subsection{Quadratic construction and conditioning}

Given an SSE instance $G$, set
\begin{equation}
    \mu=c
    \label{eq:ridge-strength}
\end{equation}
and define
\begin{equation}
    F_G(x)=x^\top L_Gx+\mu\norm{x}_2^2-2\one^\top x.
    \label{eq:quadratic-F}
\end{equation}

\begin{lemma}[Uniform condition-number bound]
\label{lem:condition-number}
The matrix $L_G+\mu I$ has spectrum contained in
$[\mu,2+\mu]$.  Consequently, every restricted condition number of a
least-squares objective with normal matrix $L_G+\mu I$ is at most
\[
    K(c)=1+\frac{2}{c}.
\]
\end{lemma}

\begin{proof}
For every $x\in\R^V$,
\[
    0\leq x^\top L_Gx
    =\sum_{u<v}w_{uv}(x_u-x_v)^2
    \leq2\sum_{u<v}w_{uv}(x_u^2+x_v^2)
    \leq2\norm{x}_2^2,
\]
where the last inequality uses \cref{eq:weighted-regularity}; self-loop
weights only reduce the total non-loop weight incident to a vertex.  Hence
$L_G$ has spectrum in $[0,2]$, as is standard for normalized
Laplacians~\cite{chung1997spectral}.  Adding $\mu I$ puts the spectrum in
$[\mu,2+\mu]$.
Every principal submatrix has eigenvalues between the global
minimum and maximum by Cauchy interlacing.  Since $\mu=c$, the ratio is
at most $(2+\mu)/\mu=1+2/c$.
\end{proof}

\subsection{Exact rational realization}

Let $B_G$ be the weighted incidence matrix having one row
$\sqrt{w_{uv}}(e_u-e_v)^\top$ for each non-loop edge $\{u,v\}$.  Its diagonal
and off-diagonal entries satisfy
\[
    (B_G^\top B_G)_{uu}
    =\sum_{v\neq u}w_{uv}=1-w_{uu},
    \qquad
    (B_G^\top B_G)_{uv}=-w_{uv}\quad(u\neq v).
\]
Hence $B_G^\top B_G=I-W_G=L_G$.  The most direct real least-squares
representation is
\[
    \widetilde A_G=
    \begin{pmatrix}
        B_G\\
        \sqrt{\mu}\,I
    \end{pmatrix},
    \qquad
    \widetilde b_G=
    \begin{pmatrix}
        0\\
        \mu^{-1/2}\one
    \end{pmatrix}.
\]
Indeed,
\[
    \norm{\widetilde A_Gx-\widetilde b_G}_2^2
    =x^\top L_Gx
      +\mu\norm{x}_2^2-2\one^\top x+\frac n\mu
    =F_G(x)+\frac n\mu,
\]
and $\widetilde A_G^\top\widetilde A_G=L_G+\mu I$.  This factorization may
contain irrational square roots of edge weights or of $\mu$.  The next lemma
replaces both blocks by rational rows without changing either identity.

\begin{lemma}[Exact rational realization]
\label{lem:rational-realization}
There is a polynomial-time construction of a rational matrix $A_G$ and a
rational vector $b_G$ such that, for every $x\in\R^V$,
\begin{equation}
    Q_G(x):=\norm{A_Gx-b_G}_2^2
    =F_G(x)+\frac{n}{\mu},
    \label{eq:least-squares-realization}
\end{equation}
and
\begin{equation}
    A_G^\top A_G=L_G+\mu I.
    \label{eq:normal-matrix}
\end{equation}
The construction has polynomial size, every row of $A_G$ has at most two
nonzero entries, and all entries have polynomial bit complexity.
\end{lemma}

\begin{proof}
We first record a deterministic rational sum-of-squares decomposition.  Given
a positive rational number $a/b$ in lowest terms, put $N=ab$.  Read the binary
expansion of $N$.  A nonzero bit $2^{2j}$ contributes the square $(2^j)^2$,
whereas a nonzero bit $2^{2j+1}$ contributes two copies of $(2^j)^2$.  Thus, in
polynomial time, we obtain
\[
    N=\sum_{t=1}^{\ell}r_t^2,
    \qquad \ell=O(1+\log N).
\]
Setting $\lambda_t=r_t/b$ gives
\begin{equation}
    \sum_{t=1}^{\ell}\lambda_t^2
    =\frac{N}{b^2}=\frac ab.
    \label{eq:rational-square-decomposition}
\end{equation}

Apply this decomposition to every positive non-loop edge weight, writing
\[
    w_{uv}=\sum_{t=1}^{\ell_{uv}}\lambda_{uv,t}^2.
\]
For each coefficient, add the row
$\lambda_{uv,t}(e_u-e_v)^\top$ to $A_G$ and response zero to $b_G$.  Together
these rows contribute exactly $w_{uv}(x_u-x_v)^2$.  Self-loops require no
rows, because they make no contribution to $L_G$.

Apply the same decomposition to $\mu$, writing
\[
    \mu=\sum_{t=1}^{\ell_\mu}\lambda_{\mu,t}^2.
\]
For every vertex $i$ and every $t\in\{1,\ldots,\ell_\mu\}$, add the row
$\lambda_{\mu,t}e_i^\top$ to $A_G$ and the response
$\lambda_{\mu,t}/\mu$ to $b_G$.  For a fixed coordinate $i$, these rows
contribute
\[
    \sum_{t=1}^{\ell_\mu}
    \left(\lambda_{\mu,t}x_i-\frac{\lambda_{\mu,t}}{\mu}\right)^2
    =\mu x_i^2-2x_i+\frac1\mu.
\]
Summing over vertices and edges proves
\cref{eq:least-squares-realization,eq:normal-matrix}.

Every edge row has two nonzero entries, while every ridge row has one.  This
proves the asserted row-sparsity property.

For a weight $a/b$, the construction uses $O(1+\log a+\log b)$ rows, and every
coefficient $r_t/b$ has $O(1+\log a+\log b)$ bits.  Summed over the explicitly
encoded nonzero weights, the number of rows and their total bit complexity are
polynomial in the graph encoding.  The ridge coefficient $\mu=c$ is fixed, so
its block has constant overhead per vertex.
\end{proof}

By \cref{lem:condition-number,lem:rational-realization}, $A_G^\top A_G$ is
positive definite and every restricted condition number of $A_G$ is at most
$K(c)$.  The additive shift $n/\mu$ affects neither minimizers nor value gaps.

\section{Completeness and soundness}
\label{sec:gap-analysis}

Fix rational constants $L\geq1$ and $\theta\in(0,1)$ satisfying
\begin{equation}
    \theta<\frac{1}{2L},
    \label{eq:theta-small}
\end{equation}
and set $\eta=\theta^2$.  Let $G$ be an instance of the resulting exact-volume
SSE promise, write $k=\delta_{\mathrm{SSE}}n$, and restrict attention to
admissible instances with $k\geq2$.  Define
\begin{align}
    Y_G
    &:=\frac{n}{\mu}-\frac{k}{\theta^2+c},
    \label{eq:gap-threshold}\\
    \Gamma(c,L,\theta)
    &:=\frac{1}{\theta^2+c}
      -\frac{\theta}{c}
      -\frac{L}{\tfrac12(1-\theta)^2+c}.
    \label{eq:gamma-gap}
\end{align}

The next theorem records the two bounds for arbitrary parameters satisfying
\cref{eq:theta-small}; \cref{sec:parameters} will choose them so that
$\Gamma(c,L,\theta)>0$.

\begin{theorem}[Reduction gap]
\label{thm:reduction-gap}
For the rational least-squares instance $Q_G$ constructed in
\cref{sec:reduction}, the following bounds hold:
\begin{enumerate}
    \item If $G$ is a YES instance, then
    \[
        \min_{\norm{x}_0\leq k}Q_G(x)\leq Y_G.
    \]
    \item If $G$ is a NO instance, then every $x$ with
    $\norm{x}_0\leq Lk$ satisfies
    \[
        Q_G(x)\geq
        Y_G+k\Gamma(c,L,\theta).
    \]
\end{enumerate}
\end{theorem}

The remainder of this section proves \cref{thm:reduction-gap}.

\subsection{Completeness}
\label{sec:completeness}

\begin{lemma}[YES-instance upper bound]
\label{lem:completeness}
If $G$ is a YES instance, then
\begin{equation}
    \min_{\norm{x}_0\leq k}F_G(x)
    \leq-\frac{k}{\theta^2+c}.
    \label{eq:yes-F-bound}
\end{equation}
Equivalently,
\begin{equation}
    \min_{\norm{x}_0\leq k}Q_G(x)\leq
    Y_G.
    \label{eq:yes-Q-bound}
\end{equation}
\end{lemma}

\begin{proof}
Let $S^\star$ be the promised set of size $k$ and expansion at most
$\theta^2$.  For $x=\alpha\one_{S^\star}$,
\[
    x^\top L_Gx
    =\alpha^2w_G(S^\star,V\setminus S^\star)
    \leq\alpha^2\theta^2k.
\]
Consequently,
\[
    F_G(\alpha\one_{S^\star})
    \leq k\bigl((\theta^2+c)\alpha^2-2\alpha\bigr).
\]
The right-hand side is minimized at
$\alpha=1/(\theta^2+c)$, where it equals the right-hand side of
\cref{eq:yes-F-bound}.  Adding $n/\mu$ proves
\cref{eq:yes-Q-bound}.
\end{proof}

\subsection{Support decomposition and Dirichlet expansion}
\label{sec:decomposition}

Fix a support $S\subseteq V$, and let $L_S$ be the principal submatrix of
$L_G$ indexed by $S$.  Directly from the Laplacian quadratic form, every
$x\in\R^S$ satisfies
\begin{equation}
    x^\top L_Sx
    =\sum_{\substack{u<v\\u,v\in S}}w_{uv}(x_u-x_v)^2
      +\sum_{v\in S}w_G(\{v\},V\setminus S)x_v^2.
    \label{eq:principal-dirichlet-form}
\end{equation}
Now fix a partition $S=U\cup W$ with $U\cap W=\varnothing$.  Each edge between
$U$ and $W$ contributes a nonnegative square to
\cref{eq:principal-dirichlet-form}.
Deleting precisely these squares separates the form into two blocks and can
only decrease it.  Let $L_{G[W]}$ denote the Laplacian of the weighted graph
induced by $W$, put
\[
    c_w=w_G(\{w\},V\setminus S),
\]
and define the modified Dirichlet Laplacian
\begin{equation}
    L'_W:=L_{G[W]}+\operatorname{diag}(c_w)_{w\in W}.
    \label{eq:modified-dirichlet-matrix}
\end{equation}
Thus $L'_W$ retains the edges inside $W$ and the Dirichlet boundary terms from
$W$ to $V\setminus S$, but not the edges between $W$ and $U$.  Equivalently,
\begin{equation}
    x^\top L'_Wx
    =\sum_{\substack{u<v\\u,v\in W}}w_{uv}(x_u-x_v)^2
      +\sum_{w\in W}c_wx_w^2,
    \label{eq:modified-dirichlet-form}
\end{equation}
For $X\subseteq W$, put
\begin{equation}
    E'_W(X)=w_G(X,W\setminus X)+w_G(X,V\setminus S)
    \label{eq:modified-boundary}
\end{equation}
and
\begin{equation}
    h'(W)=\min_{\varnothing\neq X\subseteq W}
        \frac{E'_W(X)}{|X|}.
    \label{eq:modified-expansion}
\end{equation}
Define $L'_U$ symmetrically, retaining the edges inside $U$ and the Dirichlet
boundary terms from $U$ to $V\setminus S$ while deleting the edges to $W$.
By \cref{eq:principal-dirichlet-form}, for every $x\in\R^S$,
\begin{equation}
\begin{aligned}
    x^\top L_Sx
    -x_U^\top L'_Ux_U
    -x_W^\top L'_Wx_W
    &=\sum_{u\in U}\sum_{w\in W}w_{uw}(x_u-x_w)^2
    \geq0.
\end{aligned}
    \label{eq:cross-edge-difference}
\end{equation}
Equivalently,
\begin{equation}
    L_S\succeq L'_U\oplus L'_W.
    \label{eq:deleted-cross-edges}
\end{equation}
Here $\oplus$ denotes the block-diagonal direct sum, with the $U$-coordinates
ordered before the $W$-coordinates.

The following is the standard Dirichlet form of Cheeger's
inequality~\cite{chung2007random}; we include the short proof to record the
normalization used here.

\begin{lemma}[Modified Dirichlet Cheeger inequality]
\label{lem:dirichlet-cheeger}
If $W\neq\varnothing$, then
\begin{equation}
    \lambda_{\min}(L'_W)\geq\frac12h'(W)^2.
    \label{eq:dirichlet-cheeger}
\end{equation}
\end{lemma}

\begin{proof}
Let $x\in\R^W$ be nonzero.  Replacing $x$ by $|x|$ does not increase the
quadratic form in \cref{eq:modified-dirichlet-form}, so assume $x\geq0$.
Apply Cauchy--Schwarz to two vectors indexed jointly by the internal edges and
the vertices of $W$: their edge coordinates are
$\sqrt{w_{uv}}|x_u-x_v|$ and $\sqrt{w_{uv}}(x_u+x_v)$, while both boundary
coordinates are $\sqrt{c_w}x_w$.  This gives
\begin{align*}
    &\sum_{\substack{u<v\\u,v\in W}}w_{uv}|x_u-x_v|(x_u+x_v)
      +\sum_{w\in W}c_wx_w^2\\
    &\quad\leq
      (x^\top L'_Wx)^{1/2}
      \left(
        \sum_{\substack{u<v\\u,v\in W}}w_{uv}(x_u+x_v)^2
        +\sum_{w\in W}c_wx_w^2
      \right)^{1/2}\\
    &\quad\leq (x^\top L'_Wx)^{1/2}\sqrt{2}\norm{x}_2.
\end{align*}
The last step uses $(a+b)^2\leq2(a^2+b^2)$ and
$2\deg_W(w)+c_w\leq2$, where $\deg_W(w)$ denotes the total weight from $w$
to the other vertices of $W$; this is where the degree-one normalization
\cref{eq:weighted-regularity} is used.

For $t\geq0$, let $W_t=\{w\in W:x_w>t\}$.  Since $x\geq0$,
\[
    |x_u^2-x_v^2|
    =\int_0^\infty
      \left|\mathbf 1_{\{x_u>t\}}-\mathbf 1_{\{x_v>t\}}\right|2t\,dt,
    \qquad
    x_w^2=\int_0^\infty\mathbf 1_{\{x_w>t\}}2t\,dt.
\]
Summing these identities gives the coarea formula
\begin{align*}
    &\sum_{\substack{u<v\\u,v\in W}}w_{uv}|x_u^2-x_v^2|
      +\sum_{w\in W}c_wx_w^2\\
    &\quad=\int_0^\infty E'_W(W_t)\,2t\,dt
      \geq h'(W)\int_0^\infty |W_t|\,2t\,dt
      =h'(W)\norm{x}_2^2.
\end{align*}
The left-hand side is the same expression bounded above by
Cauchy--Schwarz.  Combining the two estimates and squaring gives
\[
    \frac{x^\top L'_Wx}{\norm{x}_2^2}
    \geq\frac12h'(W)^2.
\]
Taking the infimum over nonzero $x$ proves the claim.
\end{proof}

The next lemma is the usual low-expansion peeling argument; closely related
decompositions appear, for example, in the proof of the Small-Set Expansion
algorithms of Arora, Barak, and Steurer~\cite{arora2010subexponential}.

\begin{lemma}[Greedy support decomposition]
\label{lem:greedy-decomposition}
Let $S\subseteq V$ be nonempty and let $\tau\in(0,1)$.  There is a partition
$S=U\cup W$, with $U\cap W=\varnothing$, such that
\begin{enumerate}
    \item if $W\neq\varnothing$, then $h'(W)\geq\tau$;
    \item if $U\neq\varnothing$, then $\Phi_G(U)<\tau$.
\end{enumerate}
\end{lemma}

\begin{proof}
Initialize $W=S$ and $U=\varnothing$.  While there is a nonempty set
$X\subseteq W$ with
\begin{equation}
    w_G(X,W\setminus X)+w_G(X,V\setminus S)<\tau|X|,
    \label{eq:greedy-removal-condition}
\end{equation}
move $X$ from $W$ to $U$.  The process terminates because $S$ is finite, and
the first conclusion is exactly the absence of a violating set at
termination.

If $U=\varnothing$, the second conclusion is vacuous.  Otherwise, let
$X_1,\ldots,X_m$ be the removed sets.  Summing
\cref{eq:greedy-removal-condition} over the iterations counts all the boundary
weight $w_G(U,V\setminus U)$ and possibly some additional nonnegative
contributions from edges between different removed pieces.  Hence
\[
    w_G(U,V\setminus U)<\tau\sum_{i=1}^m|X_i|
    =\tau|U|.
\]
Division by $|U|$ proves the second conclusion.
\end{proof}

We use \cref{lem:greedy-decomposition} with
\begin{equation}
    \tau=1-\theta.
    \label{eq:tau-choice}
\end{equation}
Then, whenever $W\neq\varnothing$,
\begin{equation}
    \lambda_{\min}(L'_W)
    \geq\frac12(1-\theta)^2
    \label{eq:residual-spectral-gap}
\end{equation}
by \cref{lem:dirichlet-cheeger}.

\subsection{The exact-volume core lemma}
\label{sec:core-bound}

The next lemma is the step that uses the exact-cardinality NO promise.  It
controls a set whose cardinality is not known in advance using only expansion
information at cardinality exactly $k$.

\begin{lemma}[Size of the nonexpanding core]
\label{lem:core-size}
Assume $G$ is a NO instance with promise parameter $\eta=\theta^2$ and target
cardinality $k=\delta_{\mathrm{SSE}}n\geq2$.  Let $L\geq1$ and
$\theta\in(0,1)$ satisfy $\theta<1/(2L)$, let $S\subseteq V$ satisfy
$|S|\leq Lk$, and apply \cref{lem:greedy-decomposition} with
$\tau=1-\theta$.  Then the removed core satisfies
\begin{equation}
    |U|<\theta k.
    \label{eq:core-size-bound}
\end{equation}
\end{lemma}

\begin{proof}
There is nothing to prove when $U=\varnothing$.  Otherwise,
\cref{lem:greedy-decomposition} gives
\begin{equation}
    \Phi_G(U)<1-\theta.
    \label{eq:core-low-expansion}
\end{equation}

We first prove that $|U|\leq k$.  Suppose instead that $m:=|U|>k$, and let
$T$ be a uniformly random $k$-element subset of $U$.  Recall that
$I_G(A)=w_G(A,A)=|A|(1-\Phi_G(A))$.  Each ordered pair of
distinct vertices in $U$ belongs to $T\times T$ with probability
$k(k-1)/[m(m-1)]$, while a diagonal pair $(u,u)$ belongs to $T\times T$ with
probability $k/m$, which is larger.  Therefore
\[
    \E I_G(T)
    \geq \frac{k(k-1)}{m(m-1)}I_G(U).
\]
Using \cref{eq:expansion-internal-edges},
\begin{equation}
    \E\Phi_G(T)
    \leq1-\frac{k-1}{m-1}\frac{I_G(U)}{m}.
    \label{eq:random-subset-expansion}
\end{equation}
By \cref{eq:core-low-expansion}, $I_G(U)/m>\theta$.  Moreover,
$m\leq|S|\leq Lk$ and $k\geq2$, so
\[
    \frac{k-1}{m-1}
    \geq\frac{k-1}{Lk-1}
    >\frac{1}{2L}.
\]
Consequently,
\[
    \E\Phi_G(T)<1-\frac{\theta}{2L}.
\]
Every realization of $T$ has cardinality exactly $k$, so the NO promise gives
$\E\Phi_G(T)\geq1-\theta^2$.  Since $\theta<1/(2L)$,
\[
    \E\Phi_G(T)
    <1-\frac{\theta}{2L}
    <1-\theta^2
    \leq\E\Phi_G(T),
\]
a contradiction.  Hence $|U|\leq k$.

Choose any $R\subseteq V\setminus U$ with $|R|=k-|U|$ and put
$T=U\cup R$.  The NO promise gives $\Phi_G(T)\geq1-\theta^2$, while
\[
    w_G(T,V\setminus T)
    \leq w_G(U,V\setminus U)+|R|,
\]
where the degree-one normalization \cref{eq:weighted-regularity} bounds the
total weight incident to every added vertex by one.  It follows that
\[
    1-\theta^2
    \leq\Phi_G(T)
    \leq\frac{|U|\Phi_G(U)+k-|U|}{k}.
\]
After rearranging,
\[
    |U|\bigl(1-\Phi_G(U)\bigr)\leq\theta^2k.
\]
Since \cref{eq:core-low-expansion} implies
$1-\Phi_G(U)>\theta$, we obtain $|U|<\theta k$.
\end{proof}

\subsection{Soundness}
\label{sec:soundness}

\begin{lemma}[Restricted-support minimum]
\label{lem:restricted-support-minimum}
Let $S\subseteq V$ be nonempty, and let $S=U\cup W$ with
$U\cap W=\varnothing$.  If $W\neq\varnothing$, then
\begin{equation}
    \min_{\supp(x)\subseteq S}F_G(x)
    \geq-\frac{|U|}{\mu}
          -\frac{|W|}{\lambda_{\min}(L'_W)+\mu}.
    \label{eq:restricted-support-minimum}
\end{equation}
If $W=\varnothing$, the second term is absent.
\end{lemma}

\begin{proof}
Let $L_S$ be the principal submatrix of $L_G$ indexed by $S$.  For vectors
supported on $S$,
\[
    F_G(x)=x_S^\top(L_S+\mu I)x_S-2\one_S^\top x_S.
\]
By \cref{eq:deleted-cross-edges},
$L_S\succeq L'_U\oplus L'_W$.
After adding $\mu I$, both sides are positive definite.  For positive-definite
matrices, a larger quadratic form has a smaller inverse, so
\[
    (L_S+\mu I)^{-1}
    \preceq
    (L'_U+\mu I)^{-1}\oplus(L'_W+\mu I)^{-1}.
\]
Completing the square and applying this inverse bound gives
\begin{align*}
    \min_{\supp(x)\subseteq S}F_G(x)
    &=-\one_S^\top(L_S+\mu I)^{-1}\one_S\\
    &\geq
      -\one_U^\top(L'_U+\mu I)^{-1}\one_U
      -\one_W^\top(L'_W+\mu I)^{-1}\one_W.
\end{align*}
Since $L'_U\succeq0$ and, when $W\neq\varnothing$,
$L'_W\succeq\lambda_{\min}(L'_W)I$, the corresponding inverse bounds are
\[
    (L'_U+\mu I)^{-1}\preceq\frac1\mu I,
    \qquad
    (L'_W+\mu I)^{-1}
    \preceq\frac{1}{\lambda_{\min}(L'_W)+\mu}I
    \quad(W\neq\varnothing).
\]
Taking quadratic forms with $\one_U$ and $\one_W$, and using
$\norm{\one_U}_2^2=|U|$ and $\norm{\one_W}_2^2=|W|$, gives the claimed bounds.
\end{proof}

\begin{theorem}[NO-instance lower bound]
\label{thm:soundness}
If $G$ is a NO instance, then every $x$ with $\norm{x}_0\leq Lk$ satisfies
\begin{equation}
    F_G(x)
    \geq
    -\frac{\theta k}{c}
    -\frac{Lk}{\tfrac12(1-\theta)^2+c}.
    \label{eq:no-F-bound}
\end{equation}
Equivalently,
\begin{equation}
    Q_G(x)
    \geq
    \frac{n}{\mu}
    -\frac{\theta k}{c}
    -\frac{Lk}{\tfrac12(1-\theta)^2+c}.
    \label{eq:no-Q-bound}
\end{equation}
\end{theorem}

\begin{proof}
Let $S=\supp(x)$.  If $S=\varnothing$, then $F_G(x)=0$, and
\cref{eq:no-F-bound} is immediate.  Otherwise, apply
\cref{lem:greedy-decomposition} with $\tau=1-\theta$ to obtain
$S=U\cup W$ with $U\cap W=\varnothing$.  By \cref{lem:core-size},
$|U|<\theta k$, and $|W|\leq|S|\leq Lk$.

If $W=\varnothing$, \cref{lem:restricted-support-minimum} gives
\[
    F_G(x)\geq-\frac{|U|}{\mu}
    \geq-\frac{\theta k}{c},
\]
which is stronger than \cref{eq:no-F-bound}.  If $W\neq\varnothing$, combine
\cref{lem:restricted-support-minimum,eq:residual-spectral-gap} with
$\mu=c$ to obtain
\begin{align*}
    F_G(x)
    &\geq-\frac{|U|}{\mu}
           -\frac{|W|}{\lambda_{\min}(L'_W)+\mu}\\
    &\geq-\frac{\theta k}{c}
           -\frac{Lk}{\tfrac12(1-\theta)^2+c}.
\end{align*}
Adding the shift $n/\mu$ proves \cref{eq:no-Q-bound}.
\end{proof}

\begin{proof}[Proof of \cref{thm:reduction-gap}]
The YES bound is \cref{lem:completeness}.  In the NO case,
\cref{thm:soundness} and the definitions in
\cref{eq:gap-threshold,eq:gamma-gap} give
\[
    Q_G(x)
    \geq Y_G+k\Gamma(c,L,\theta)
\]
for every $\norm{x}_0\leq Lk$.
\end{proof}

\section{Parameter selection and hardness}
\label{sec:main-proof}

\subsection{Choosing the constants}
\label{sec:parameters}

Fix constants $\gamma\in(0,1]$ and $C\geq1$.  We choose constants in the order
\[
    (\gamma,C)\ \longrightarrow\ c\ \longrightarrow\ L
    \ \longrightarrow\ \theta\ \longrightarrow\
    \eta=\theta^2\ \longrightarrow\ \delta_{\mathrm{SSE}}.
\]
The constant $c$ controls the ridge strength, while $L$ bounds the output
support blow-up.

\begin{lemma}[Parameter selection]
\label{lem:parameter-selection}
There exist rational constants $c\in(0,1)$, $L\geq1$, and
$\theta\in(0,1)$ such that
\[
    \Gamma(c,L,\theta)>0
    \qquad\text{and}\qquad
    \theta<\frac{1}{2L}.
\]
Moreover, if the design matrix $A$ of an instance satisfies
\begin{equation}
    \kappa_r(A)\leq K(c):=1+\frac{2}{c}
    \qquad\text{for every }r,
    \label{eq:uniform-kappa-bound}
\end{equation}
then the output of a $(C,\gamma)$ algorithm has support at most $Lk$ on every
run on which its support guarantee holds.
\end{lemma}

\begin{proof}
Let
\[
    L_0(c)=\max\left\{1,
        C\left(1+\frac{2}{c}\right)^{1-\gamma}\right\}.
\]
As $c\to0^+$,
\[
    cC\left(1+\frac{2}{c}\right)^{1-\gamma}
    =Cc^\gamma(2+c)^{1-\gamma}\longrightarrow0.
\]
Since also $c\to0$, it follows that $c(L_0(c)+1)\to0$.  We may therefore
choose a rational $c\in(0,1)$ so small that
$c(L_0(c)+1)<1/2$.  Choose a rational $L$ with
$L_0(c)\leq L<L_0(c)+1$.  In particular, $cL<1/2$.

At $\theta=0$, the first term in \cref{eq:gamma-gap} is $1/c$, whereas
the last term is $L/(1/2+c)$.  Since $cL<1/2<1/2+c$, we have
\[
    \frac1c>\frac{L}{1/2+c}.
\]
By continuity and the density of the rationals, a sufficiently small rational
$\theta>0$ makes $\Gamma(c,L,\theta)$ positive and satisfies
\cref{eq:theta-small}.

Finally, suppose that \cref{eq:uniform-kappa-bound} holds and that $x$ is the
output on a run on which the support guarantee in
\cref{eq:algorithm-definition} holds.  With $r=\norm{x}_0$, we have
\[
    r\leq Ck\kappa_{r+k}(A)^{1-\gamma}
      \leq CkK(c)^{1-\gamma}
      \leq L_0(c)k
      \leq Lk.
\]
\end{proof}

For the rest of the paper, fix $c,L,\theta$ as in
\cref{lem:parameter-selection}, set $\eta=\theta^2\in(0,1/2)$, and invoke
\cref{hyp:exact-sse} with this value of $\eta$.  These constants are fixed for
the chosen $(C,\gamma)$ and are hard-wired into the reduction.  Write
$k=\delta_{\mathrm{SSE}}n$ and restrict attention to sufficiently large
admissible instances, so in particular $k\geq2$.

Since $\delta_{\mathrm{SSE}}\leq\eta=\theta^2$ and
$\theta<1/(2L)$,
\[
    L\delta_{\mathrm{SSE}}
    \leq L\theta^2<1,
\]
and hence $Lk<n$.  Thus the dense zero-residual solution $\mu^{-1}\one$ lies
outside the support range of the NO-instance lower bound.

\subsection{Proof of the main theorem}

Define the positive gap
\begin{equation}
    \Delta_G=k\Gamma(c,L,\theta)>0,
    \label{eq:decision-gap}
\end{equation}
where $\Gamma$ is from \cref{eq:gamma-gap}.  By
\cref{thm:reduction-gap}, the stated YES upper bound and NO lower bound are
separated by $\Delta_G$.

Set
\begin{equation}
    \eps_G=\frac{\Delta_G}{3}.
    \label{eq:accuracy-choice}
\end{equation}

\begin{proof}[Proof of \cref{thm:main-intro}]
Suppose for contradiction that an algorithm ruled out by
\cref{thm:main-intro} exists.  It is a $(C,\gamma)$ algorithm in the sense of
\cref{def:algorithm}.  Given a graph $G$ from the promise problem fixed above,
the finitely many smaller input sizes are handled directly by enumerating their
$k$-element subsets, so assume that $n$ lies in the sufficiently large regime.
Construct $(A_G,b_G)$ as in \cref{lem:rational-realization}.  By
\cref{lem:condition-number,lem:rational-realization}, $A_G$ has full column
rank, every row of $A_G$ has at most two nonzero entries, and
\[
    \kappa_r(A_G)\leq K(c)
    \qquad\text{for every }r.
\]
Because $c$ was fixed as a function only of $(C,\gamma)$, this proves the
structural assertion in the theorem with $K_{C,\gamma}=K(c)$.
Call the algorithm on
\[
    (A_G,b_G,k,\eps_G),
\]
let $x$ be the returned vector, and let $\mathcal E$ be the event that both
guarantees in \cref{eq:algorithm-definition} hold.  By the assumed success
guarantee,
\begin{equation}
    \Pr(\mathcal E)\geq\frac23.
    \label{eq:success-event}
\end{equation}
Hence, on $\mathcal E$, \cref{lem:parameter-selection} gives
$\norm{x}_0\leq Lk$.

Recall $Y_G$ from \cref{eq:gap-threshold}.  
Output YES if and only if
\begin{equation}
    Q_G(x)<Y_G+\frac{2\Delta_G}{3};
    \label{eq:decision-rule}
\end{equation}
otherwise, output NO.
If $G$ is a YES instance, then on $\mathcal E$ the additive guarantee and
\cref{thm:reduction-gap} give
\[
    Q_G(x)\leq\min_{\norm{z}_0\leq k}Q_G(z)+\eps_G
    \leq Y_G+\frac{\Delta_G}{3},
\]
so the rule outputs YES.  If $G$ is a NO instance, then on $\mathcal E$ the
support bound allows us to apply \cref{thm:reduction-gap}, yielding
\[
    Q_G(x)\geq Y_G+\Delta_G,
\]
so the rule outputs NO.  Thus the decision rule is correct whenever
$\mathcal E$ occurs.

The constructed instance has a polynomial-size rational encoding by
\cref{lem:rational-realization}.  Since $c,L,\theta$ are fixed rational
constants, all quantities in the decision rule also have polynomial-size
rational encodings.  Moreover, $k\geq1$ and $\Gamma(c,L,\theta)>0$, so
\[
    \frac1{\eps_G}
    =\frac{3}{k\Gamma(c,L,\theta)}
    \leq\frac{3}{\Gamma(c,L,\theta)}.
\]
Thus the call to the hypothetical algorithm remains polynomial-time.  Its
rational output $x$ has encoding length bounded by its running time, so
$\norm{x}_0$ and $Q_G(x)$ can be computed exactly in polynomial time.  The
decision rule would therefore solve the SSE promise in randomized polynomial
time with success probability at least $2/3$, contradicting
\cref{hyp:exact-sse}.
\end{proof}

\begin{corollary}[Hardness at fixed additive accuracy]
\label{cor:fixed-accuracy}
Fix any rational constant $\eps_0>0$.  Under \cref{hyp:exact-sse}, the
conclusion of \cref{thm:main-intro} continues to hold even for algorithms
specialized to the single accuracy value $\eps=\eps_0$.
\end{corollary}

\begin{proof}
Since
\[
    \Delta_G
    =k\Gamma(c,L,\theta)
    =\delta_{\mathrm{SSE}}n\Gamma(c,L,\theta),
\]
we have $\Delta_G>3\eps_0$ for all sufficiently large $n$ in the hard family.
On those instances, repeat the proof of \cref{thm:main-intro}, calling the
hypothetical algorithm at its fixed accuracy $\eps_0$.  In the YES case,
\[
    Q_G(x)
    \leq Y_G+\eps_0
    <Y_G+\frac{\Delta_G}{3},
\]
whereas the NO-case bound is unchanged.  Hence the same decision rule
distinguishes the two cases, contradicting \cref{hyp:exact-sse}.
\end{proof}

\section*{Acknowledgements}

The proof was first obtained using a fully automated Gemini-based agentic
system developed internally at Google. The authors have verified the proof and
edited it for clarity of presentation, and take responsibility for the final
version. The authors would like to thank Kyriakos Axiotis for helpful
discussions.
\bibliographystyle{alpha}
\bibliography{references}

\end{document}